\documentclass[11pt]{article}
    \usepackage[margin=2cm]{geometry}
    \usepackage{times,bbm,microtype}
    \usepackage{amsmath,amsthm,amssymb,stmaryrd,xspace,comment,booktabs, enumerate}
    \usepackage[T1]{fontenc}
    
    \usepackage[dvipsnames]{xcolor}
    \usepackage{algorithm}
    \usepackage[noend]{algpseudocode}
    \usepackage{thmtools,thm-restate}
    \usepackage{url}
    \usepackage[sort]{cite}
    \usepackage[linktocpage=true,
	pagebackref=true, colorlinks, linkcolor=BrickRed, citecolor=blue]{hyperref}
    \usepackage{cleveref}
    \usepackage{tikz} 
    \usepackage{caption}
    \usepackage[framemethod=tikz]{mdframed}
    \mdfdefinestyle{thmstyle}{
  backgroundcolor=gray!10,
  linewidth=0pt,
  roundcorner=4pt,
  innertopmargin=\topsep,
  innerbottommargin=\topsep,
  skipabove=\baselineskip,
  skipbelow=\baselineskip
}

    \title{The Power of the Score Sequence of a Tournament}
    \author{Prantar Ghosh\footnote{(pghosh@tntech.edu) Department of Computer Science, Tennessee Tech University, USA. Funded in part by startup support provided by the Department of Computer Science and the College of Engineering at Tennessee Tech. Part of this work was done while the author was visiting the University of Birmingham, UK, hosted by Sagnik Mukhopadhyay.} \and Sahil Kuchlous\footnote{(sahilkuchlous@gmail.com) Work done in part while the author was an undergraduate student at Harvard University, USA.} \and Shravan Mehra\footnote{(shravansinghmehra@gmail.com) School of Computer Science, University of Birmingham, UK} \and Sagnik Mukhopadhyay\footnote{(schwagznikst@gmail.com), School of Computer Science, University of Birmingham, UK. Partially supported by UKRI New Investigator Award EP/X03805X/2.}
    }

    \date{}
    \usepackage{amsmath,amsthm,amssymb,stmaryrd,xspace,comment,booktabs}
 \usepackage{thmtools, mdframed}
 \mdfdefinestyle{thmstyle}{
  backgroundcolor=gray!10,
  linewidth=0pt,
  roundcorner=4pt,
  innertopmargin=\topsep,
  innerbottommargin=\topsep,
  skipabove=\baselineskip,
  skipbelow=\baselineskip
}

\newcommand{\xqedhere}[2]{\rlap{\hbox to#1{\hfil\llap{\ensuremath{#2}}}}}

\newcommand{\eps}{\varepsilon}

\newcommand{\ola}{\textsf{OLA}}

\newcommand{\NN}{\mathbb{N}}

\newcommand{\bx}{\mathbf{x}}

\newcommand{\cC}{\mathcal{C}}

\newcommand{\cF}{\mathcal{F}}
\newcommand{\cG}{\mathcal{G}}

\newcommand{\cP}{\mathcal{P}}

\newcommand{\cU}{\mathcal{U}}

\newcommand{\tO}{\widetilde{O}}
\newcommand{\tE}{\widetilde{E}}

\newcommand{\mypar}[1]{\medskip\noindent{\bfseries #1.}~}

\newtheorem{theorem}{Theorem}[section]
\newtheorem{lemma}[theorem]{Lemma}

\newtheorem{corollary}[theorem]{Corollary}

\newtheorem{fact}[theorem]{Fact}

\newtheorem{remark}[theorem]{Remark}
\theoremstyle{definition}

\renewcommand{\geq}{\geqslant}

\DeclareMathOperator{\poly}{poly}
\DeclareMathOperator{\polylog}{polylog}

\newcommand{\din}{d_{\text{in}}}
\newcommand{\dout}{d_{\text{out}}}

\newcommand{\prantar}[1]{
}
\newcommand{\sahil}[1]{
}
\newcommand{\shravan}[1]{
}
\newcommand{\sagnik}[1]{
}

\newcommand{\inCut}{\overleftarrow{\cC}}

\newcommand{\outCut}{\overrightarrow{\cC}}
\newcommand{\Ubar}{\overline{U}}
\newcommand{\Fbar}{\overline{F}}

\newcommand{\dem}{\mathsf{dem}}
\newcommand{\skel}{\mathsf{S}}
\newcommand{\degs}{\mathbf{d}}
\newcommand{\SFM}{\mathsf{SFM}}

\begin{document}

   \maketitle
   
    \begin{abstract}
\textit{What problems can one solve on a tournament if only its score sequence is known?}

Tournaments are oriented complete graphs that form an extensively-studied class of directed graphs (digraphs), both from combinatorial and algorithmic perspectives. Over the years, researchers have identified multiple classical digraph problems that can be solved on a tournament from only its score sequence (indegree sequence). These problems include {\em acyclicity testing} and {\em topological sorting} [Chakrabarti, Ghosh, McGregor, and Vorotnikova; SODA'20], {\em $s,t$-reachability}, {\em strong connectivity}, and decomposition into {\em strongly connected components (SCC)} [Ghosh and Kuchlous; ESA'24], and vertex-ordering problems such as {\em cutwidth} and {\em optimal linear arrangement} [Barbero, Paul, and Pilipczuk; ICALP'17]. These prior works showed the sufficiency of the score sequence by designing distinct algorithms for the individual problems. In this work, we give a simple unified framework that solves all these problems using only indegrees and, in fact, completely characterises the class of problems that is determined by the indegree information: problems whose answers are invariant under cycle reversals.  

As a byproduct of our results, we obtain algorithms for a variety of connectivity-based, cut-based, and vertex-ordering problems on tournaments and ``almost tournaments'' in the streaming, the two-player communication, and the cut-query models of computation. Some of these algorithms match existing optimal bounds and others provide bounds improving the state of the art. Specifically, our polynomial-time algorithms for almost tournaments improve upon the exponential-time algorithms of Ghosh and Kuchlous and have much simpler analyses.  

The said characterisation is a special case of a much more general result that we establish: for any {\em arbitrary} digraph, the knowledge of its skeleton (underlying undirected graph) and the vertex indegrees completely determines its properties that are invariant under cycle reversal. In particular, this gives us an $O(n^2)$-cut-query algorithm to solve directed minimum cut on $n$-node graphs in polynomial time, a significant result that has been observed in the literature but not concretely stated in this form.
Our results also unveil interesting general connections between two well-studied sublinear models for graph problems: semi-streaming and cut-query. 
\end{abstract}

    \newpage

    \newpage
\section{Introduction}\label{sec:intro}

Tournaments are directed graphs (digraphs) where each pair of nodes shares exactly one directed edge. They have a rich literature that dates back more than half a century (see the book on tournaments by Moon \cite{Moon1968TopicsOT}). In particular, the {\em score sequence} of a tournament, which is the sequence of its indegrees, has been a topic of extensive study \cite{ ClaessonDF23, Kleitman1981ForestsAS,  Moon_1962,Narayana_Bent_1964, RIORDAN1968, stockmeyer2023counting}. It has been noted sporadically in the literature that multiple classical digraph problems can be solved on a tournament from only its indegree sequence. For instance, Chakrabarti, Ghosh, McGregor, and Vorotnikova \cite{ChakrabartiGMV20} observed that the acyclicity of a tournament and the topological ordering of an acyclic tournament are completely determined by its score sequence: for the former, check if the indegree sequence is a permutation of $\{0,\ldots, n-1\}$, and for the latter, sort the nodes by indegree. Again, it is known that sorting by indegree yields optimal solutions for two widely-studied vertex ordering problems: {\em cutwidth} (as shown by Fradkin \cite{Fradkin-thesis11}) and {\em optimal linear arrangement} (shown by Barbero, Paul, and Pilipczuk~\cite{BarberoPP17}); see \Cref{sec:probs} for the definitions of these problems. A recent and more surprising result of Ghosh and Kuchlous \cite{GhoshK24} says that the indegree information also determines $s,t$-reachability, strong connectivity, and, in general, all strongly connected components (SCCs) of a tournament. This raises an interesting combinatorial question: 
\vspace{-1mm}

\begin{center}
{\em Can we characterise the class of problems determined by the score sequence of a tournament?}    
\end{center}

\vspace{-1mm}
    
We answer this question in the affirmative and show that this class is {\em precisely} the set of all digraph problems whose answers are invariant under cycle reversal, i.e., whose answers remain the same even if we flip any number of cycles in the graph. Then, given only the score sequence, we always have an algorithm to solve such a problem, even if at the cost of exponential runtime: brute force over all possible tournaments on the given number of nodes to find one satisfying the given score sequence and solve the problem on that input using a classical algorithm. Can we, however, always guarantee a {\em polynomial-time} algorithm? We answer this in the affirmative as well, provided that the problem admits a polynomial-time solution in the classical model. Indeed, we achieve this by designing a polynomial-time (in fact, almost-linear-time) algorithm to generate a tournament satisfying a given indegree sequence. We thus have the following theorem.

\vspace{-2mm}

\begin{mdframed}[style=thmstyle]
 \begin{restatable}{theorem}{indegt} \label{thm:tour-comb}
    The class of graph problems determined by the indegree sequence of a tournament is precisely the class of problems whose answers are invariant under cycle reversal. Moreover, each problem can be determined with an additive overhead of almost-linear time ($n^{2+o(1)}$ time for an $n$-node tournament) on its classical time-complexity.
\end{restatable}
 \end{mdframed}

 \vspace{-1mm}

Since all problems mentioned in the first paragraph can be easily shown to belong to this class, our result generalises all of them. To solve those problems using the indegree information, prior work proposed different algorithms to process the indegrees, tailored according to the respective problems. \Cref{thm:tour-comb} gives a simple \textit{unified} polynomial-time framework for all those problems since all of them admit polynomial-time algorithms in the classical setting where the entire input can be accessed.

\mypar{Generalisation for Arbitrary Digraphs} Since tournaments form a rather narrow subclass of digraphs, it is natural to ask: {\em can we generalise the result to a broader class of graphs?} \Cref{thm:tour-comb} is indeed a special case of our more general theorem: all problems whose answers are invariant under cycle reversal can be solved on any general digraph only from the knowledge of its underlying undirected graph a.k.a.~{\em skeleton} and indegree sequence. 

\vspace{-2mm}

\begin{mdframed}[style=thmstyle]
 \begin{restatable}{theorem}{maincomb} \label{thm:gendig-comb}
     The class of graph problems determined by the skeleton and indegree sequence of a directed graph is precisely the class of problems whose answers are invariant under cycle reversal. Moreover, each problem can be determined with an additive overhead of almost-linear time ($m^{1+o(1)}$ time for an $m$-edge graph) on its classical time-complexity.
\end{restatable}
 \end{mdframed}

 \vspace{-1mm}

Since the skeleton of a tournament is the complete graph, its score sequence suffices.

\mypar{Applications: Sublinear Algorithms}  In the modern world of massive graphs, a natural goal is to analyse them by extracting as little information as possible from the graphs. Tournaments, which have the maximum possible edges, are prime examples of such graphs. They appear in the real world when, for example, there is a pairwise ranking or preference (represented by directed edges) among all elements in a dataset \cite{GassTournaments, GengLQALS08, Rosti-etal-2007-combining}. Common questions that arise here include detecting whether the rankings are consistent (acyclicity testing), and if so, ordering the elements in increasing order of preference (topological sorting), whether an element $t$ is (directly or indirectly) preferred over another element $s$ ($s,t$-reachability), or partitioning the elements into groups such that there is a total ordering of preference among the groups (SCC decomposition). 

This motivates the study of {\em streaming algorithms} for tournaments \cite{ BawejaJW22, ChakrabartiGMV20, ChenCLT26Independence, GhoshK24}. These algorithms make one or a few passes over the edges while maintaining a small summary and then process the summary to generate an output. They usually aim for {\em semi-streaming} space, i.e., $\tO(n):= O(n\poly(\log n))$ space for $n$-node graphs.\footnote{Most graph problems turn out to be intractable when restricted to smaller memory \cite{feigenbaumkmsz05}.} Observe that problems completely determined by the vertex indegrees admit single-pass semi-streaming algorithms: while processing the stream, we just keep count of the indegrees using $O(n\log n)$ space. By \Cref{thm:tour-comb}, we thus get semi-streaming algorithms for all cycle-reversal-invariant problems on tournaments, and they also carry over to a couple of other popular models for sublinear graph algorithms: (a) the {\em cut-query model} where we solve a problem by optimising queries to an oracle that returns the (directed) cut size for any subset of nodes, i.e., the number of edges incoming to that subset, and (b) the {\em two-player communication model} where the input graph is split between two players who need to communicate (possibly in multiple rounds of interaction) as few bits as possible to solve the problem. For (a), we need $n$ (singleton) cut queries to find all vertex indegrees. For (b), it suffices for one player, say Alice, to communicate to the other player, say Bob, $O(n\log n)$ bits that encode the indegrees restricted to her edges, since Bob can then compute the score sequence of the entire tournament by adding the indegrees induced by his edges. Hence, we have the following theorem.

\vspace{-2mm}

\begin{mdframed}[style=thmstyle]
 \begin{restatable}{theorem}{cycinvgen}\label{thm:cycinv}
    Given any digraph problem $\cP$ whose answer is invariant under cycle reversal, there is a single-pass $O(n\log n)$-space streaming algorithm, an $O(n)$-cut-query algorithm, and an $O(n\log n)$-bit communication protocol for the problem $\cP$ on $n$-node tournaments. Each algorithm's runtime has an additive $n^{2+o(1)}$-factor overhead on the runtime of a classical algorithm for $\cP$.
\end{restatable}
\end{mdframed}

\vspace{-2mm}

In fact, Chakrabarti et al.~\cite{ChakrabartiGMV20} and Ghosh and Kuchlous \cite{GhoshK24} discovered their aforementioned results while studying streaming algorithms for tournaments and consequently obtained single-pass semi-streaming algorithms for acyclicity testing, topological sorting, $s,t$-reachability, strong connectivity, and SCC decomposition. Our result implies such algorithms for all these problems, and also the first single-pass semi-streaming algorithms for the fundamental problems of directed global-minimum cut and minimum $s,t$-cut on tournaments. Notably, although there is a substantial body of work on \textit{undirected} min-cut in the streaming and communication models \cite{ApersEGLMN22, AssadiCK19, AssadiD21, BlikstadBEMN22, JiangNS26, Kenneth-MordochK26STOC,  MukhopadhyayN20, RubinsteinSW18, Zelke11}, the complexity of directed (global or $s,t$) min-cut was unknown, even for tournaments. In particular, we obtain the following corollary.

\vspace{-1mm}

\begin{restatable}{corollary}{tourimps}\label{cor:algs}
    There is a single-pass $O(n\log n)$-space streaming algorithm, an $O(n)$-cut-query algorithm, and an $O(n\log n)$-bit communication protocol for the following problems on $n$-node tournaments, each of which runs in polynomial time.
    \begin{enumerate}[(i)]
    \item Connectivity-based problems: $s,t$-reachability, strong connectivity, SCC decomposition (or more generally finding the SCC-graph);
    \item DAG-based problems: acyclicity testing, topological sorting;
    \item Vertex ordering problems: optimal linear arrangement (\ola), cutwidth;
    \item Directed cut problems: global minimum cut, minimum $s,t$-cut.
    \end{enumerate}

\noindent
    All these bounds are optimal (up to a $\log n$ factor).
\end{restatable}

 We remark that having the same ``sketch'' of the tournament---the indegree-count---for a large variety of problems is very convenient from a practical perspective: the sketch is deterministic, can be computed very easily and quickly, and rather than storing separate summaries based on the respective problems at hand, we can just post-process this single sketch accordingly upon receiving the respective queries.  

\mypar{Extension to Almost Tournaments} Again, a natural question about generalisability arises: {\em can we extend these streaming algorithms to a broader class of graphs?} Indeed, we cannot fully generalise them since most of these problems are hard in streaming (i.e., require $\Omega(n^2)$ space for a single pass) for general digraphs \cite{ChakrabartiGMV20, feigenbaumkmsz05}. However, Ghosh and Kuchlous \cite{GhoshK24} showed that they are tractable when the digraph is an ``almost tournament'': for digraphs {\em $k$-close to a tournament}, i.e., digraphs that can be obtained by deleting at most $k$ edges from a tournament, they gave single-pass $\tO(n+k)$-space streaming algorithms for the connectivity-based problems mentioned above. Our general theorem for arbitrary digraphs (\Cref{thm:gendig-comb}) implies these algorithms and extends them to any cycle-reversal-invariant problem: we can compute the skeleton of such a graph using $O(k\poly(\log n))$ space since, as shown by the said work, we can recover the $k$ non-edges using this space. Then, standard streaming-to-communication simulations imply $\tO(n+k)$-bit communication protocols for any such problem. Again, the idea can be extended to the cut-query model to obtain similar bounds: the $k$ non-edges can be retrieved using $O(k \log n)$ cut queries by standard techniques. We thus get the following corollary.

\begin{restatable}{corollary}{kcloseimps}\label{cor:kclose}
     Given any digraph problem $\cP$ whose answer is invariant under cycle reversal, there is a single-pass $\tO(n+k)$-space streaming algorithm, an $\tO(n+k)$-cut-query algorithm, and an $\tO(n+k)$-communication protocol for $\cP$ on $n$-node digraphs that are $k$-close to tournaments. Each algorithm's runtime has an additive $m^{1+o(1)}$-factor overhead on the runtime of a classical algorithm for $\cP$, where $m$ is the number of edges in the digraph.
\end{restatable}

Furthermore, we get an \textit{exponential improvement in runtime} over \cite{GhoshK24}'s algorithm for almost tournaments. Their runtime for $s,t$-reachability and strong connectivity is $\tO(2^{\min(k, n)})$ and that for SCC decomposition or SCC-graph construction is $\tO(2^n)$. Our algorithm is not only polynomial time for any $k$, but it also has much simpler analysis, especially for the SCC decomposition algorithm where their analysis is somewhat intricate.

\mypar{Other Applications} We remark that our general theorem (\Cref{thm:gendig-comb}) for arbitrary digraphs implies the following results in the cut-query model.

\begin{mdframed}[style=thmstyle]
 \begin{restatable}{corollary}{quadquery} \label{cor:quadcutq}
     For any directed graph, global- and $s,t$-minimum cut can be solved with $O(n^2)$ directed cut queries and polynomial runtime.
\end{restatable}
 \end{mdframed}

 Indeed, the directed global-minimum cut problem in the cut-query model is a special instance of the fundamental submodular function minimisation ($\SFM$) problem since the cut function (directed or undirected) is submodular.\footnote{A function $f : 2^\cU \to \mathbb{R}$ is \textit{submodular} iff for 
any $S, T \subseteq \cU$, we have
$f(S) + f(T) \geq f(S \cup T) + f(S \cap T)$.} The $\SFM$ problem asks to minimise a submodular function $f : 2^\cU \to \mathbb{R}$ when given access to an oracle that, upon query $S\subseteq \cU$, returns the value of $f(S)$. There has been a long line of work \cite{AxelrodLS20, ChakrabartyGJS22, ChakrabartyGJS23, ChakrabartyLS0W19, Cun85, Cunningham85, FleischerI03, GLS1988, Iwata03,  IwataFF01, IwataO09, Jiang23, LeeSW15, Orlin09, Schrijver00, Vygen03} to pin down the exact complexity of $\SFM$. The current best algorithm due to Jiang \cite{Jiang21} makes $\tO(n^2)$ queries (here, $n:=|\cU|$) but has exponential runtime. The best-known (strongly) polynomial-time algorithm has query complexity $\tO(n^3)$, also by \cite{Jiang21}. Therefore, these algorithms for $\SFM$ do not imply \Cref{{cor:quadcutq}}. Again, while the undirected analogue of \Cref{{cor:quadcutq}} is trivial \cite{RubinsteinSW18}, the directed version is not; since unlike undirected graphs, it is provably impossible to recover a directed graph exactly using $O(n^2)$ cut queries.\footnote{In fact, the impossibility holds for any number of cut queries; to wit, a directed cycle and its reversal are distinct labelled graphs, but all cut queries have identical answers on both of them.} \Cref{cor:quadcutq}, however, follows from some observations in prior work (\cite{Cun85}, Section 2) but was not stated in this form, which we believe is worth stating explicitly in light of the discussion above. 

Finally, in \Cref{app:conn}, we show that our results establish an interesting connection between the streaming and the cut-query models. For tournament problems, the existence of a cut-query algorithm (irrespective of the number of queries) implies the existence of a semi-streaming algorithm (\Cref{lem:qry-solve-semi-str}). Thus, strong streaming lower bounds can be a useful tool for proving intractability in the cut query model. Again, since this implies that the semi-streaming model is at least as (algorithmically) strong as the cut-query model for tournament problems, a natural question is whether it is \textit{strictly} stronger: is there a tournament problem that is undecidable by cut queries, but solvable by a semi-streaming algorithm? We answer this question in the affirmative (\Cref{prop:qrystrsep}).

\subsection*{Related Work and Comparisons}


\mypar{Streaming and Communication} 
Apart from the results mentioned above, Chakrabarti et al.~\cite{ChakrabartiGMV20} obtained algorithms and lower bounds for approximate feedback arc set in tournaments (FAST) and for source/sink-finding, both in adversarial-order and random-order streams. Baweja, Jia, and Woodruff \cite{BawejaJW22} then improved \cite{ChakrabartiGMV20}'s results for FAST. They also gave $(p+1)$-pass $O(n^{1+1/p})$-space algorithms for SCC decomposition (also implying an algorithm for reachability), strong connectivity, and Hamiltonian path on tournaments. Ghosh and Kuchlous \cite{GhoshK24} improved upon their results for the connectivity-based problems by designing single-pass semi-streaming algorithms for each of them. They also gave new algorithms for Hamiltonian cycle and FAST, and established lower bounds for multiple problems including (exact and approximate) FAST, $s,t$-distance, source/sink finding, and the connectivity and DAG-based problems. All these works considered the respective tournament problems in the two-party communication model in order to establish lower bounds. Mande, Paraashar, Sanyal, and Saurabh~\cite{MandePSS24} recently studied the communication complexity of multiple tournament problems including king finding, source finding, and maximum degree, giving both upper and lower bounds.

The connectivity- and DAG-based problems on \textit{general digraphs} have been studied in the streaming and communication models by multiple works \cite{AssadiJJST22, AssadiR20, BawejaJW22, ChakrabartiGMV20, ChenKPS0Y21, feigenbaumkmsz05, GuruswamiO13}. These works mostly studied lower bounds for the problems, with the best-known ones being $\Omega(n^2)$ space for a single pass, $\Omega(n^{2-o(1)})$ space for $O(\sqrt{\log n})$ passes, and $\Omega(n^{1+\Omega(1/p)})$ space for $p$ passes, thereby needing at least $\Omega(\log n/\log \log n)$ passes for semi-streaming algorithms. The  $s,t$-reachability problem is the only one among them (to the best of our knowledge) that has a non-trivial pass upper bound for semi-streaming algorithms:  $O(n^{1/2+o(1)})$ passes suffice \cite{AssadiJJST22}. On arbitrary graphs under W-streams (allowing outputs to be streamed), Laura and Santaroni \cite{lauraS11} studied SCC decomposition, obtaining an $O(n)$-pass $O(n\log n)$-space algorithm.

\mypar{Cut query} Graur et al.~\cite{GraurPRW20} showed that a directed graph can be fully learned up to cycle reversals via $O(n^2)$ cut queries,\footnote{We came to know via personal communication that an independent unpublished work by Chakrabarty, Chen, and Khanna \cite{ChakrabartyCK22} also obtained a similar result.} which implies exponential-time $O(n^2)$-cut-query algorithms for directed global- and $s,t$-minimum cut. They, however, did not mention a concrete polynomial-time algorithm that generates such a digraph, which would have then implied our \Cref{cor:quadcutq}. Cunningham \cite{Cun85} proved a statement similar to \Cref{cor:quadcutq} by a reduction to the min-cut problem on an appropriately constructed digraph, but the construction is very different from ours. We are not aware of any prior work that studied tournaments in the cut-query model.    
The \textit{undirected} cut-query model has been recently studied extensively \cite{ApersEGLMN22, BlikstadBEMN22, ChakrabartyL23, GraurPRW20, JiangNS26, Kenneth-MordochK25ESA, Kenneth-MordochK26SODA, Kenneth-MordochK26STOC, LeeSW15, MukhopadhyayN20,  RubinsteinSW18}.

\mypar{Other settings} Frank and Gy\'{a}rf\'{a}s \cite{FrankGyarfas1978} studied the problem of orienting the edges of an undirected graph so as to satisfy target upper and lower bounds on the vertex in- and out-degrees, a generalised version of our main problem (\Cref{lem:main}). They mentioned an algorithm for this problem that flips directed paths and also remarked that the problem can be solved using network flows, but did not specify the algorithm. While they might have alluded to an algorithm similar to ours, the final result we get for general digraphs (\Cref{thm:gendig-comb}) is novel. We also note that Coppersmith et al.~\cite{CoppersmithFR10} proved that the indegree sequence of a tournament suffices to achieve a $5$-approximation to FAST. The problems of optimal linear arrangement (\ola) and cutwidth have been widely studied from the perspective of computational complexity as well as fixed-parameter tractability (FPT) \cite{BarberoPP17, bfis25, devos2025maximumlineararrangementproblem, fp13, Fradkin-thesis11}. 

    \section{Preliminaries}\label{sec:prelims}

\mypar{Notation and Terminology} We state some notation and terminology that we use throughout the paper. The notation $\tO(f)$ hides factors polylogarithmic in $f$. For $a\in \NN$, we have $[a]:=\{1,\ldots,a\}$. Unless otherwise specified, a general digraph is denoted by $G = (V, E)$ where $|V| = n$ and $|E| = m$. The {\em skeleton} of a digraph is its underlying undirected graph, and is denoted by $\skel = (V,\tE)$ where $\tE := \{\{u, v\}: (u, v)\in E\}$ is a multiset—allowing $\skel$ to contain up to two parallel edges (corresponding to directed $2$-cycles in $G$). 
Given a subset of edges $F\subseteq \tE$ in the skeleton of $G$, the notation $F[G]$ denotes the directed subgraph of $G$ induced by $F$ (i.e., edges in $F$ oriented as in $G$). 
For a vertex $v$, we use $\din(v, G)$ and $\dout(v, G)$ to denote the indegree and outdegree of $v$ respectively in the digraph $G$. We simply write $\din(v)$ and $\dout(v)$ if the graph $G$ is clear from the context.
The vector $\degs$ is an $n$-length sequence of non-negative integers and may denote the indegree sequence of a graph, given by $\langle \din(1),\ldots, \din(n)\rangle$ where the vertex set is $[n]$. We use $\cG[\skel, \degs]$ to denote the equivalence class of all graphs that have skeleton $\skel$ and indegree sequence $\degs$. We say that two graphs $G_1$ and $G_2$ are related, i.e., $G_1 \sim G_2$ iff both of them have the same skeleton and indegree sequence. 
Unless otherwise specified, a {\em cut} refers to a directed cut. 
Given a vertex subset $U$ of a digraph, the cut size of $U$, i.e., the number of edges incoming to $U$, is given by $\inCut_G(U):=|\{(v,u) \in E(G) : u\in U\}|$. We often refer to cut sizes simply as ``cuts''. For a directed graph $G$, we define its {\em cut profile} as the set $\{(U,\inCut(U)): U\subseteq V\}$.  The number of edges outgoing from $U$ is given by  $\outCut_G(U):=|\{(u,v) \in E(G) : u\in U\}|$.

\subsection{Models}\label{sec:models}

We describe the algorithmic models that appear in this work. 

\mypar{The Directed Cut-Query Model} The {\em undirected} cut-query model was introduced by Rubinstein, Schramm, and Weinberg~\cite{RubinsteinSW18}. Here, we define the analogous {\em directed} cut-query model \cite{ChakrabartyCK22, GraurPRW20}. We are given query access to an oracle that holds a digraph $G=(V,E)$, where we know $V$ but not $E$. Upon a query $U\subseteq V$, the oracle returns $\inCut_G(U)$, the number of edges incoming to $U$ in $G$. The goal is to solve a problem on $G$ by making as few queries as possible.

Equivalently, we could define the model as returning $\outCut_G(U)$, the number of edges outgoing from $U$, but since $\inCut_G(U) = \outCut_G\left(V\setminus U\right)$, WLOG and incurring at most a factor of $2$ in the number of queries, we can assume that we have access to just one type of queries. We drop the subscript $G$ if the graph is clear from the context.

\mypar{The Graph Streaming Model} In this paper, we only focus on the insert-only streaming setting and define only this variant. Here, we have an input digraph $G=(V,E)$, where the vertex set $V=[n]$ and the elements of the edge set $E$ arrive one by one. The input graph and the stream order are chosen adversarially (worst case). The goal is to make one or a few passes over the stream elements, store as few bits of information as possible, and after the final pass, run a post-processing algorithm on the stored data to generate an output. The goal is to minimise the space usage and the number of passes. The semi-streaming model \cite{feigenbaumkmsz05} restricts the space usage to $\tO(n)$. In this work, we only deal with single-pass algorithms.

\mypar{Two-Party Communication} We study graph problems in the two-player communication model of Yao \cite{yao79}. Here, players Alice and Bob both know the vertex set $V$ and receive disjoint edge sets $E_A, E_B \subseteq V\times V$ respectively. They send messages to each other (possibly over multiple rounds) to compute an element of $\cF(G)$, where $\cF$ is a  predetermined relation defined on all graphs on $V$, and $G$ is the graph $(V, E_A\sqcup E_B)$. The \emph{communication cost} of a protocol for $\cF$ is defined as the maximum number of bits exchanged by the players, over all possible inputs $(E_A,E_B)$. If the protocol is randomised, then the maximum is also taken over all possible random strings used by the players. The {\em communication complexity} of a relation $\cF$ is the minimum communication cost over all protocols that return a valid output for $\cF$ on all inputs. If the protocol is randomised, we require the answer to be correct on all inputs with probability at least $2/3$ over the randomness of the players' random strings. All our protocols in this paper are deterministic and single-round, i.e., Alice sends a single message to Bob, following which he returns the output. Our lower bounds hold even for the stronger model allowing randomness and unlimited rounds of interaction.

\subsection{Basic Tools}\label{sec:tools}

\begin{fact}[Sparse recovery \cite{gilbertI10, DasV13}]\label{fact:sp-rec}
    Given streaming updates to a vector $\bx\in [-M,M]^N$, there is a deterministic $O(k\cdot\polylog (M,N))$-space algorithm that, at the end of the stream, recovers $\bx$ in polynomial time if $\bx$ has at most $k$ non-zero entries.
\end{fact}

\begin{fact}[Undirected Graph Recovery \cite{RubinsteinSW18}]\label{fact:graph-req-cq}
    Given cut-query access to an undirected graph $G$ on $n$ nodes and $m$ edges, there is a deterministic $O(\min(m\log n, n^2))$-query algorithm that recovers $G$.
\end{fact}

\subsection{Problems}\label{sec:probs}

We formally define some problems considered in this paper. The input to each of them is a digraph $G = (V,E)$.

\begin{itemize}
    \item \textit{$s,t$-Reachability:} Given $s,t\in V$, is there a directed path from $s$ to $t$?
    \item \textit{Strong connectivity:} Does each node have a directed path to every other node?
    \item \textit{SCC decomposition:} Partition the vertex set $V$ into $V_1, \ldots, V_\ell$ such that each $V_i$ induces a strongly connected component (SCC) of $G$, i.e., a maximal subgraph of $G$ that is strongly connected.
    \item \textit{Acyclicity testing:} Is there a directed cycle in $G$?
    \item \textit{Topological sorting:} Given that $G$ is a DAG, output an ordering of $V$ such that for each edge $(u,v)\in E$, $u$ appears (somewhere) before $v$ in the ordering.
    \item \textit{Global minimum cut:} Find the size of the smallest directed cut, i.e., $\min_{U\subseteq V} \inCut(U)$.
    \item \textit{Minimum $s,t$-cut:} Find the size of the smallest directed cut whose removal makes $s$ unreachable from $t$, i.e., $\min \{\inCut(U): U\subseteq V, s\in U, t\not\in U \}$.
    \item For an ordering $\sigma$ of the digraph nodes, the {\em width} of $\sigma$ is defined as the maximum cut size over the prefixes of $\sigma$. The \textit{cutwidth} of the digraph is the minimum width over all vertex orderings, i.e., 
    $$\text{cutwidth}(G):=\min_{\sigma\in S_n} \left\{\max_{\text{prefix $P$ of }\sigma} \inCut(P) \right\}.$$
    \item \textit{Optimal Linear Arrangement} (\ola) is a vertex ordering that minimises the sum (rather than maximum) of the prefix-cut sizes, i.e.,
$$\mathsf{OLA}(G):=\text{argmin}_{\sigma\in S_n} \left\{\sum_{\text{prefix $P$ of }\sigma}\inCut(P)  \right\}.$$
    Equivalently, it is an ordering that minimises the sum of the {\em stretches} of the induced back edges \cite{BarberoPP17}.

    For each of cutwidth and \ola, we refer to the problem version that asks for the value and the one that asks for the ordering as its ``value version'' and ``ordering version'' respectively.
\end{itemize}

    \section{Graph Characterisation and Computation Based on Skeleton and Indegrees}

 Recall that $\cG[\skel,\degs]$ is the equivalence class of all graphs with skeleton $\skel$ and indegree sequence $\degs$. Two graphs $G_1$ and $G_2$ are related, i.e., $G_1 \sim G_2$, iff both have the same skeleton and indegree sequence, i.e., both belong to the same equivalence class. We prove the following general lemma that implies our main theorems.

\begin{mdframed}[style=thmstyle]
\begin{lemma}\label{lem:main}
\begin{enumerate}
\item[(1)] For any two digraphs $G_1$ and $G_2$, we have $G_1 \sim G_2$ iff there is a set of cycles in $G_1$, flipping the directions of which yields $G_2$.

    \item[(2)] Given an undirected graph $\skel$ on $n$ nodes and $m$ edges and a sequence $\degs$ of $n$ non-negative integers, we have an $O(m^{1+o(1)})$-time
    deterministic algorithm that outputs a directed graph from $\cG[\skel,\degs]$ if $\cG[\skel,\degs]$ is non-empty, and otherwise outputs $\bot$.

\end{enumerate}
\end{lemma}
\end{mdframed}

\mypar{Characterisation} We first prove \Cref{lem:main} (1), which characterises graphs based on skeleton and indegree sequence.

\begin{proof}[Proof of \Cref{lem:main} (1)]
    The ``if'' direction is straightforward. Given any $v\in V$, each cycle in $G_1$ containing $v$ has one edge outgoing from $v$ and one edge incoming to $v$. Therefore, reversing a cycle preserves the indegree of vertex $v$. Also, reversing the direction of a cycle does not change the skeleton of the graph. Hence, if $G_2$ is obtained by reversing cycles in $G_1$, the graphs $G_1$ and $G_2$ must have the same skeleton and indegree sequence, i.e., $G_1\sim G_2$.

    For the ``only if'' direction, $G_1$ and $G_2$ have the same skeleton, say $\skel$. Let $F$ be the set of (undirected) edges in $\skel$ that have opposite directions in $G_1$ and $G_2$. Let $\Fbar:= \skel\setminus F$. Fix any vertex $v\in V$.  We have
     \begin{align}
     \label{eq:indegGprime}
        \din(v, G_2) 
        = \din(v, F[G_2]) + \din(v, \Fbar [G_2])
        = \dout(v, F[G_1]) + \din(v,  \Fbar[G_1])
     \end{align}
     The second equality follows since precisely the edges in $F$ have opposite directions in $G_1$ and $G_2$, and the edges in $\Fbar$ have the same direction in $G_1$ and $G_2$. Again,
    \begin{align}
        \label{eq:indegG}\din(v, G_1) 
        = \din(v, F[G_1]) + \din(v,  \Fbar[G_1])
     \end{align}
    But since $G_1$ and $G_2$ have the same indegree sequence, $\din(v,G_1)=\din(v,G_2)$. Hence, from \cref{eq:indegGprime} and \cref{eq:indegG}, we get $\dout(v, F[G_1]) = \din(v, F[G_1])$. Since this is true for all $v\in V$, $F[G_1]$ is an Eulerian circuit, a collection of edge-disjoint cycles. We can obtain $G_2$ from $G_1$ by precisely reversing $F[G_1]$. 
\end{proof}

Therefore, if a graph property is invariant under cycle reversal, then the skeleton and the indegree sequence suffice to determine it.

\mypar{Completeness of characterisation} We note that if a property is {\em not} invariant under cycle reversal, then the skeleton and the indegrees alone cannot decide the property. Assume to the contrary that there is a property $\cP$ that is not invariant under cycle reversal, but there is an algorithm that decides it from only the skeleton and the indegree information. By definition of $\cP$, there must exist graphs $G$ and $G'$ such that $G'$ is obtained by reversing cycles of $G$, but $\cP$ has different answers on $G$ and $G'$. By \Cref{lem:main} (1), $G\sim G'$. Therefore, they have the same skeleton, say $\skel$, and the same indegree sequence, say $\degs$. Our algorithm outputs an answer that is a function of only $\skel$ and $\degs$. Thus, it outputs the same answer for $\cP$ on both $G$ and $G'$, and hence, must be incorrect on one of them. This shows that the class of cycle-reversal-invariant properties {\em completely} characterises the set of properties determined by the skeleton and the indegrees.

\begin{corollary}\label{cor:maincombnec}
    The set of well-defined properties for the equivalence under skeleton and indegree sequence is \textit{exactly} the class of properties that are invariant under reversal of cycles.
\end{corollary}

\begin{remark}
    The characterisation can be equivalently stated with the flipped cycles being ``edge-disjoint.'' This is because the effect of reversing two overlapping cycles $C_1$ and $C_2$ with $C_1\cap C_2 = F$ can be obtained by reversing the cycle $(C_1\cup C_2) \setminus F$.  
\end{remark}




\mypar{Efficient computation} \Cref{lem:main} (1) only implies an exponential-time algorithm for any cycle-reversal-invariant problem when we are given the skeleton and indegree sequence of the input digraph: we brute force over all orientations of the given skeleton and find one that satisfies the given indegree sequence. The invariance guarantees that solving the problem on this digraph would give us the solution for the original digraph. Now we focus on generating such a feasible digraph efficiently in {\em polynomial time}.


\begin{proof}[Proof of \Cref{lem:main} (2)]
We describe a polynomial-time algorithm to assign directions to edges in $\skel$ so as to obtain a digraph $G'\in \cG[\skel, \degs]$. 
Let $w(u, v)$ be the number of edges between $u$ and $v$ in $\skel$. If $w(u, v) = 2$, $G'$ must have both $\overrightarrow{uv}$ and $\overrightarrow{vu}$. Therefore, we can assign directions to all such edges.
   
    Ignore all edges whose directions have been marked. Therefore, for the remaining edges $\{u,v\}$ in $\skel$, we have $w(u, v) =0 \text{ or }1$. 
    Now, we find a flow $f$ which satisfies the following equations:
    \begin{align}
        \forall \{u,v\}\in E(\skel) &:~~ f(u,v)=-f(v,u)\label{const:flow1}\\
        \forall \{u,v\}\in E(\skel) &:~~|f(u, v)| \leq w(u, v) &&w(u, v) = 0 \text{ or }1\label{const:flow2}\\
        \forall u\in V &: \sum_{v\in N_{\skel}(u)} f(u, v) = d_{out}(u) - d_{in}(u)\label{const:flow3}
    \end{align}
    where $E(\skel)$ is the set of edges in $\skel$, and $N_\skel(u)$ is the neighbourhood of $u$ in the graph $\skel$.
    
    We check for a feasible flow satisfying these constraints. This is the well-studied \textit{circulation problem with vertex demands} \cite{GoldbergT90, Tardos85, Tarjan91} and  can be solved using any polynomial-time $(s,t)$-maxflow algorithm in the following way: let us define the demand of a vertex $u$ to be $\dem(u) := d_{out}(u) - d_{in}(u)$. We create a source node $s$ and a sink node $t$. For each node $u$ with $\dem(u) > 0$, we put an edge from $s$ to $u$ with capacity $\dem(u)$. Similarly, for every node $u$ with $\dem(u)< 0$, we put an edge from $u$ to $t$ with capacity $-\dem(u)$. Note that this augmentation modifies the flow conservation constraints to ``$\forall u\in V: \sum_{v\in N_{\skel}(u)} f(u, v) = 0$'', which means that we can now run an $(s,t)$-max-flow algorithm on this augmented graph, and the corresponding flow on the original graph will maintain all original constraints.  Using the recent breakthrough result of \cite{Brand0PKLGSS23}, we can solve it deterministically in $O\left(m^{1+o(1)}\right)$ time. 
    
    If there is no feasible flow, then we return $\bot$, announcing that $\cG[\skel, \degs] = \emptyset$. To argue correctness, observe that if $\cG[\skel, \degs] \neq \emptyset$ and $G=(V,E)\in \cG[\skel, \degs]$, then the following flow is feasible.  
    \[
f(u,v) =
\begin{cases}
1 & \text{if } (u,v) \in E, \\
-1 & \text{if } (v,u) \in E,
\end{cases}
\]
Thus, we shall always find a feasible flow in this case. The contrapositive says that if we do not find a feasible flow, then $\cG[\skel, \degs]$ must be empty. 
    
    Suppose that we do find a feasible flow. Then we describe a strategy to orient the remaining edges to construct a graph $G'$ that we output. As $w(u, v)$ and $\dem(u)$ are integers, we know that the flow $f$ is also integral (therefore, $f(u, v) > 0 \implies f(u, v) = 1$). Now, to all edges $\{u, v\}\in \skel$ such that $f(u, v) = 1$, we assign the direction $\overrightarrow{uv}$. Let $\skel'\subseteq \skel$ be the edges that have not been assigned a direction yet, i.e., the edges $\{u,v\}\in \skel$ that have flow $f(u,v)=0$. Because the demand constraint of each node $v\in V$ has been satisfied, we know that $\din(v, \skel'[G']) = \dout(v, \skel'[G'])$ for all $v\in V$. This is equivalent to the fact that the subgraph $\skel'[G']$ forms an Eulerian circuit. We can now find an Euler tour on $\skel'$ or detect if none exists in $O(m)$ time  \cite{fleischner1990eulerian}; if $\skel'$ is not Eulerian, then we know that $\cG[\skel, \degs] = \emptyset$. If we do find an Euler tour, then orienting the edges in $\skel'$ along the direction of the tour, we ensure that for every vertex $v$, the in- and out-degree of $v$ restricted to edges in $\skel'$ are equal. As a result, we manage to assign a direction to each edge in $\skel$ to obtain a digraph $G'$ that has the indegree sequence $\degs$. The total runtime is $O\left(m^{1+o(1)}\right)+O(m) = O\left(m^{1+o(1)}\right).$
    \end{proof}

The proof of \Cref{lem:main} is now complete. From \Cref{lem:main} and \Cref{cor:maincombnec}, we have our main theorem.

\begin{mdframed}[style=thmstyle]
    \maincomb*
\end{mdframed}

\mypar{Applications to Cut Query} Since for any subset of nodes $U$ of a digraph, the cut size $\inCut(U)$ is invariant under cycle reversal, we have the following corollary of \Cref{thm:gendig-comb} that is significant from the perspective of cut-query algorithms.

\begin{corollary}\label{cor:cut-invariance}
   Given an undirected graph $\skel$ and a sequence $\degs$ of non-negative integers, the cut profiles of all graphs in $\cG[\skel, \degs]$ are identical. 
\end{corollary}

Although it follows immediately from \Cref{thm:gendig-comb}, we give an independent direct and simple proof of the statement.

\begin{proof}
If $\cG[\skel, \degs] =\emptyset$, then this is vacuously true. Hence, we assume that $\cG[\skel, \degs] \neq \emptyset$. Let $G$ be any graph in $ \cG[\skel, \degs]$. For any $U\subseteq V$, let $\cC_{\skel}(U)$ denote the undirected cut size of $U$ in $\skel$. Observe that 
    \begin{align}
    &\outCut_G(U) + \inCut_G(U) = \cC_{\skel}(U),\label{eq:lin1}\\
    &\outCut_G(U) - \inCut_G(U) = \sum_{v\in U} (\dout(v) - \din(v)).\label{eq:lin2} 
    \end{align}
\Cref{eq:lin2} follows since each edge $(u,w)$ inside $U$ (i.e., both $u,w\in U$) contributes to both $\dout(u)$ and $\din(w)$ and cancels. Now, for each $U\subseteq V$, the values $\inCut_G(U)$ and $\outCut_G(U)$ are determined by the linearly independent equations (\ref{eq:lin1}) and (\ref{eq:lin2}). The RHS of \cref{eq:lin1,eq:lin2} are completely determined by $\skel$ and $\degs$. Hence, for fixed $\skel$ and $\degs$, $\inCut_G(U)$ and $\outCut_G(U)$ are identical for all $G\in \cG[\skel, \degs]$. 
\end{proof}

We restate and prove \Cref{cor:quadcutq} that captures an important implication as discussed in \Cref{sec:intro}.

\begin{mdframed}[style=thmstyle]
\quadquery*
\end{mdframed}

\begin{proof}
    For an input digraph $G$, we find the indegree profile $\degs$ using $n$ queries. We recover its skeleton $\skel$ using $O(n^2)$ queries (\Cref{fact:graph-req-cq}). 
    Then using \Cref{lem:main} (2), we generate $G' \sim G$ in polynomial time and then find the global- or $s,t$-min-cut of $G'$ using any polynomial-time classical algorithm for the problem. By \Cref{cor:cut-invariance}, the answer for $G$ must be the same.
\end{proof}

\section{Applications to Algorithms for Tournaments and Almost Tournaments}

Recall that a {\em tournament} graph is one where each pair of vertices shares exactly one directed edge. In other words, a tournament is an oriented complete graph. The indegree sequence of a tournament is also called its {\em score sequence}. We obtain the following characterisation of tournaments based on their score sequences.

\begin{corollary}[score sequence characterisation]\label{cor:tour-indeg}
    Two tournaments $T$ and $T'$ have the same indegree sequence iff $T$ can be obtained from $T'$ by only reversing a set of cycles of $T'$.  
\end{corollary}

\begin{proof}
    Since $T$ and $T'$ have the same skeleton (the complete graph), the statement immediately follows from \Cref{lem:main} (1).
\end{proof}

Again, we have the following implication of \Cref{thm:gendig-comb}.

\begin{mdframed}[style=thmstyle]
\indegt*
\end{mdframed}

\Cref{thm:tour-comb} now implies our algorithmic results for tournaments in various computational models. 

\begin{mdframed}[style=thmstyle]
\cycinvgen*
\end{mdframed}

\begin{proof}
   Given an input tournament, we can find its indegree sequence in the cut-query model using $n$ singleton cut queries; in the streaming model, indegree count takes $O(n\log n)$ space and a single pass; in the communication model, Alice can send Bob $O(n \log n)$ bits that encode the indegree sequence induced by her edges $E_A$, from which Bob can compute the indegree sequence induced by all edges $E_A\sqcup E_B$. Then, by \Cref{thm:tour-comb}, we can determine the problem $\cP$ with an additive $n^{2+o(1)}$-factor overhead on the runtime of a classical algorithm that solves it. 
\end{proof}

\tourimps*

\begin{proof}
It is easy to see that reachability, strong connectivity, or in general, any strongly connected component of a graph does not change when cycles in the graph are flipped. Since a DAG has no cycle, the DAG-based properties indeed remain intact. For the vertex ordering problems and cut problems, we note that they are all certain functions of the cuts of a graph. {\em All} cuts in the graph are preserved under cycle reversal, and hence all these properties are, in turn, preserved. Therefore, all these problems are invariant under cycle reversal, and since they all have polynomial-time classical algorithms on tournaments, the upper-bound results follow from \Cref{thm:cycinv}.

The optimality of the bounds up to a factor of $\log n$ in each model---including streaming and cut query---follows from $\Omega(n)$ communication complexity of the problems. By standard streaming-to-communication simulations, an $S$-bit communication lower bound for a problem implies an $S$-space single-pass lower bound for the same problem in the streaming model. Again, an $\Omega(n/\log n)$ cut-query lower bound for each problem follows from its $\Omega(n)$ communication lower bound: a $C$-cut-query algorithm for the problem implies an $O(C\log n)$-communication protocol for it, since the players can communicate to each other the answer for each query using $O(\log n)$ bits.

We now justify the $\Omega(n)$ communication lower bound for each problem. \cite{GhoshK24} showed such lower bounds for the connectivity-based and the DAG-based problems. The lower bounds on the value versions of the vertex-ordering problems follow via trivial reductions from the acyclicity problem: the cutwidth or the \ola-value of a tournament is zero if and only if it is acyclic. Along similar lines, the lower bounds for their ordering versions follow from the topological ordering problem: the cutwidth-ordering and \ola~of an acyclic tournament are also valid topological orderings of the tournament. 

Again, the lower bounds for global- and $s,t$-minimum-cut follow via easy reductions from the strong connectivity and $s,t$-reachability problems respectively: observe that in any digraph, the minimum cut size is non-zero if and only if the graph is strongly connected. Again, the minimum $s,t$-cut size is non-zero if and only if $s$ is reachable from $t$. 
\end{proof}

Define a digraph to be {\em $k$-close to tournament} if it is obtained by deleting at most $k$ edges from a tournament. \cite{GhoshK24} showed that these graphs are some generalisations of tournaments that still admit sublinear-space streaming algorithms for connectivity-based problems. Our results imply a much more general statement.

\kcloseimps*

\newcommand{\Sbar}{\overline{\skel}}

\begin{proof}
    For a digraph that is $k$-close to tournament, its $k$ non-edges can be recovered deterministically in the single-pass streaming model using $O(k\cdot \text{polylog }n)$ space via sparse recovery (\Cref{fact:sp-rec}), and the degree information takes an additional $O(n\log n)$ space. This streaming algorithm can be easily simulated in the communication model along standard lines: Alice runs the algorithm on her edges and sends Bob the algorithm's memory state at the end of her stream. Given this state, Bob continues the stream with his edges and finds the solution for the entire graph at the end of the stream. The total communication is at most the space usage of the algorithm, which is $\tO(n+k)$. 
    
   Now consider the cut-query model. Let $\skel$ be the skeleton of the input graph $G$. We recover the complement graph $\Sbar$ of $\skel$. Note that $\Sbar$ is an undirected graph on at most $k$ edges. For each $U\subseteq V$, we need $2$ queries to the directed cut oracle on $G$ to recover the undirected cut $\cC_{\Sbar}(U)$ since $\cC_{\Sbar}(U) = |\Ubar| - \inCut_{G}(U) - \inCut_{G}(\Ubar)$, where $\Ubar:=V\setminus U$. Thus, by \Cref{fact:graph-req-cq}, we need $O(k\log n)$ directed cut queries on $G$ to recover $\Sbar$. Once we do that, we can compute $\skel$. The indegree sequence is obtained by $n$ additional queries. The results then follow from \Cref{thm:gendig-comb}. 
\end{proof}

This gives us an improvement over a prior result of \cite{GhoshK24}. They gave $\tO(2^{\min(k, n)})$-time algorithms for reachability and strong connectivity, and $\tO(2^n)$-time algorithms for SCC decomposition and constructing the SCC-graph. This is because they go over every possible ``completion'' of the almost tournament and on each completion, call a blackbox for tournaments. \Cref{cor:kclose} improves these runtimes to polynomial time -- in fact, almost-linear $O(m^{1+o(1)})$ time since each of these problems has a classical $O(m+n)$-time algorithm for any digraph.

    \section{Connections Between the Streaming and the Cut-Query Models}\label{app:conn}

Finally, we observe the following interesting connection between the semi-streaming and cut-query models with respect to tournament problems.

\begin{restatable}{lemma}{qryimplystr}\label{lem:qry-solve-semi-str}
    If a tournament problem $\cP$ is decidable/solvable by cut queries (regardless of the number of queries), then it has a single-pass $O(n\log n)$-space algorithm. More generally, if a problem $\cP$ on a digraph that is $k$-close to tournament, is solvable by cut queries, then it has a single-pass $\tO(n+k)$-space streaming algorithm.
\end{restatable}

\begin{proof}
    Given a tournament $T$, we can store its indegrees in a single pass using $O(n\log n)$ space. By \Cref{lem:main} (2), this would let us generate in polynomial time a tournament $T'$ in the same equivalence class as $T$. Then \Cref{cor:cut-invariance} says that $T'$ has the same cut profile as $T$, and so we can now access all cut queries on $T$ for free. 

    The ``more general'' statement follows along similar lines since we can recover the skeleton of a digraph that is $k$-close to tournament in $\tO(n+k)$ space (as described in the proof of \Cref{cor:kclose}).
\end{proof}

The contrapositive is also interesting: a lower bound ruling out semi-streaming algorithms for a tournament problem would imply its undecidability in the directed cut-query model.

\mypar{A separation between semi-streaming and cut queries} 
By \Cref{lem:qry-solve-semi-str}, the single-pass semi-streaming model is at least as strong as the cut query model for tournament problems. Can we show that it is {\em strictly} stronger? That is, can we exhibit a tournament problem that is solvable in single-pass semi-streaming but not in the cut-query model? The answer is yes, as long as the model allows randomness.

Recall that the FAST problem asks to find the minimum Feedback Arc Set in a Tournament, a set of arcs whose removal deletes all cycles. We are interested in the FAST-size.

\begin{restatable}{proposition}{separation}\label{prop:qrystrsep}
     The problem of finding a $(1+\eps)$-approximation to FAST-size has a single-pass semi-streaming (randomized) algorithm, but even a $o(n)$-approximation is not possible using cut queries. 
\end{restatable}

\begin{proof}
The single-pass $(1+\eps)$-approximation semi-streaming algorithm was given by \cite{ChakrabartiGMV20}. 

For the impossibility, let $T$ be the tournament on vertex set $[n]$ defined as follows. Add the edge $n\rightarrow 1$. For the remaining pairs $(i,j)\in [n]\times [n]$ with $i<j$, add the edge $i\rightarrow j$. Let $T'$ be the tournament obtained by reversing the cycle $1\rightarrow 2 \rightarrow\ldots \rightarrow n \rightarrow 1$ in $T$. Hence, $T$ and $T'$ have the same cut profile and any cut-query algorithm must output the same answer for both of them.

    In $T$, the FAST-size is $1$ (deleting the edge $n\rightarrow 1$ makes it acyclic). But in $T'$, the FAST-size $=\Omega(n)$ since $T'$ has $\Omega(n)$ edge-disjoint triangles: for each odd $i\leq n-2$, we have the triangle $i\rightarrow i+2 \rightarrow i+1 \rightarrow i$. Hence, $\alpha$-approximate FAST-size for $\alpha=o(n)$ is undecidable by cut queries.  
\end{proof}

    \section{Conclusions \& Future Directions}

 In this paper, we characterised the class of all digraph problems decidable by only its skeleton and indegrees, and consequently the class of all tournament problems decidable by only its score sequence. We also developed a polynomial-time algorithm to generate a feasible digraph satisfying a given skeleton and indegree sequence. We then exhibited a variety of applications of these results in several domains such as streaming algorithms, cut-query algorithms, and communication complexity. The complete characterisation should be influential in future research on these properties, from either an algorithmic or a combinatorial point of view.

One of the immediate future directions is to show that directed minimum cut can be solved by $\tO(n)$ cut queries or, towards hardness, requires $\Omega(n^2)$ queries. Either result would be a breakthrough as directed minimum cut is a great candidate for proving a superlinear $\SFM$ lower bound---an outstanding question that has remained open for the last five decades \cite{blikstad21, ChakrabartyGJS23, chakrabarty2019faster, Harvey08}. Because cut queries cannot fully specify edge directions, we conjecture that any improvement of the form $O(n^{2-\epsilon})$ will require designing a minimum-cut algorithm that does not crucially use the individual edge directions. This will be a departure from the standard combinatorial ways of solving this problem and will be of interest.

Another intriguing future direction is to design streaming algorithms for tournament problems that are {\em not} invariant under cycle reversal. We showed that one can obtain the entire cut profile of a tournament in semi-streaming space. Combining this strong tool with other useful sketching and sampling techniques known in the streaming literature, can we resolve the streaming complexities of important tournament problems that remain open? A couple of these problems are (approximate) FAST-ordering (an ordering of nodes that minimises the number of induced back edges) and finding Hamiltonian paths and cycles. For the former problem, there is a single-pass semi-streaming $(1+\eps)$-approximation algorithm by \cite{ChakrabartiGMV20} that takes exponential time. \cite{BawejaJW22} gave a polynomial-time algorithm with the same space and approximation factor but with $O(\log n)$ passes. The best-known polynomial-time single-pass semi-streaming algorithm achieves a $5$-approximation simply by sorting with respect to indegrees \cite{CoppersmithFR10}. Can we improve upon this approximation factor in polynomial time and a single pass?

Finally, we remark that the complexity of each digraph problem studied in this paper is unsettled for general digraphs in the multipass streaming model. In fact, even after a series of works \cite{AssadiJJST22, AssadiR20, ChakrabartiGMV20, ChenKPSSY21-lognpass, GuruswamiO13} on reachability, there is an exponential gap in the pass-complexity of semi-streaming algorithms for the problem, and it is one of the biggest open problems in multipass streaming (see the survey \cite{AssadiMultipassSurvey}). Our theorem for arbitrary digraphs provides new insights into the problem, and coupled with sophisticated ideas, might lead to improved multipass algorithms for $s,t$-reachability and related problems on general digraphs.

\section*{Acknowledgements}

We thank Deeparnab Chakrabarty for helpful discussions and for pointing us to relevant prior works by Cunningham \cite{Cun85, Cunningham85}. We are also grateful to an anonymous reviewer for referring us to the paper \cite{FrankGyarfas1978} by Frank and Gy\'{a}rf\'{a}s. Finally, we thank anonymous ESA 2026 reviewers for their suggestions and comments that helped in improving the presentation of the paper.

\bibliography{refs, biblio, references}

@article{Tardos85,
  author       = {{\'{E}}va Tardos},
  title        = {A strongly polynomial minimum cost circulation algorithm},
  journal      = {Comb.},
  volume       = {5},
  number       = {3},
  pages        = {247--256},
  year         = {1985}
}

@article{Tarjan91,
  author       = {Robert E. Tarjan},
  title        = {Efficiency of the Primal Network Simplex Algorithm for the Minimum-Cost
                  Circulation Problem},
  journal      = {Math. Oper. Res.},
  volume       = {16},
  number       = {2},
  pages        = {272--291},
  year         = {1991}
}

@article{GoldbergT90,
  author       = {Andrew V. Goldberg and
                  Robert E. Tarjan},
  title        = {Finding Minimum-Cost Circulations by Successive Approximation},
  journal      = {Math. Oper. Res.},
  volume       = {15},
  number       = {3},
  pages        = {430--466},
  year         = {1990}
}

@inproceedings{ChakrabartyGJS22,
  author    = {Deeparnab Chakrabarty and
               Andrei Graur and
               Haotian Jiang and
               Aaron Sidford},
  title     = {Improved Lower Bounds for Submodular Function Minimization},
  booktitle = {{FOCS}},
  pages     = {245--254},
  publisher = {{IEEE}},
  doi       = {10.1109/FOCS54457.2022.00030},
  year      = {2022}
}

@inproceedings{ApersEGLMN22,
  author    = {Simon Apers and
               Yuval Efron and
               Pawel Gawrychowski and
               Troy Lee and
               Sagnik Mukhopadhyay and
               Danupon Nanongkai},
  title     = {Cut Query Algorithms with Star Contraction},
  booktitle = {{FOCS}},
  pages     = {507--518},
  publisher = {{IEEE}},
  doi       = {10.1109/FOCS54457.2022.00055},
  year      = {2022}
}

@inproceedings{blikstadBEMN22,
  author    = {Joakim Blikstad and
               Jan van den Brand and
               Yuval Efron and
               Sagnik Mukhopadhyay and
               Danupon Nanongkai},
  title     = {Nearly Optimal Communication and Query Complexity of Bipartite Matching},
  booktitle = {{FOCS}},
  pages     = {1174--1185},
  publisher = {{IEEE}},
  doi       = {10.1109/FOCS54457.2022.00113},
  year      = {2022}
}

@inproceedings{blikstad21,
  author    = {Joakim Blikstad},
  title     = {Breaking O(nr) for Matroid Intersection},
  doi       = {10.4230/LIPIcs.ICALP.2021.31},
  booktitle = {{ICALP}},
  series    = {LIPIcs},
  volume    = {198},
  pages     = {31:1--31:17},
  publisher = {Schloss Dagstuhl - Leibniz-Zentrum f{\"{u}}r Informatik},
  year      = {2021}
}

@inproceedings{chakrabarty2019faster,
  author    = {Deeparnab Chakrabarty and
               Yin Tat Lee and
               Aaron Sidford and
               Sahil Singla and
               Sam Chiu{-}wai Wong},
  title     = {Faster Matroid Intersection},
  doi       = {10.1109/FOCS.2019.00072},
  booktitle = {{FOCS}},
  pages     = {1146--1168},
  publisher = {{IEEE} Computer Society},
  year      = {2019}
}

@inproceedings{AssadiJJST22,
  author    = {Sepehr Assadi and
               Arun Jambulapati and
               Yujia Jin and
               Aaron Sidford and
               Kevin Tian},
  title     = {Semi-Streaming Bipartite Matching in Fewer Passes and Optimal Space},
  booktitle = {{SODA}},
  pages     = {627--669},
  publisher = {{SIAM}},
  year      = {2022}
}

@inproceedings{AssadiCK19,
  author    = {Sepehr Assadi and
               Yu Chen and
               Sanjeev Khanna},
  title     = {Polynomial pass lower bounds for graph streaming algorithms},
  booktitle = {{STOC}},
  pages     = {265--276},
  publisher = {{ACM}},
  year      = {2019}
}

@inproceedings{ChenKPS0Y21,
  author    = {Lijie Chen and
               Gillat Kol and
               Dmitry Paramonov and
               Raghuvansh R. Saxena and
               Zhao Song and
               Huacheng Yu},
  title     = {Almost optimal super-constant-pass streaming lower bounds for reachability},
  booktitle = {{STOC}},
  pages     = {570--583},
  publisher = {{ACM}},
  year      = {2021}
}

@inproceedings{RubinsteinSW18,
  author    = {Aviad Rubinstein and
               Tselil Schramm and
               S. Matthew Weinberg},
  title     = {Computing Exact Minimum Cuts Without Knowing the Graph},
  booktitle = {{ITCS}},
  series    = {LIPIcs},
  volume    = {94},
  pages     = {39:1--39:16},
  publisher = {Schloss Dagstuhl - Leibniz-Zentrum f{\"{u}}r Informatik},
  doi       = {10.4230/LIPIcs.ITCS.2018.39},
  year      = {2018}
}

@inproceedings{LeeSW15,
  author    = {Yin Tat Lee and
               Aaron Sidford and
               Sam Chiu{-}wai Wong},
  title     = {A Faster Cutting Plane Method and its Implications for Combinatorial
               and Convex Optimization},
  booktitle = {{FOCS}},
  pages     = {1049--1065},
  publisher = {{IEEE} Computer Society},
  doi       = {10.1109/FOCS.2015.68},
  year      = {2015}
}

@inproceedings{MukhopadhyayN20,
  author    = {Sagnik Mukhopadhyay and
               Danupon Nanongkai},
  _editor    = {Konstantin Makarychev and
               Yury Makarychev and
               Madhur Tulsiani and
               Gautam Kamath and
               Julia Chuzhoy},
  title     = {Weighted min-cut: sequential, cut-query, and streaming algorithms},
  booktitle = {{STOC}},
  pages     = {496--509},
  publisher = {{ACM}},
  doi       = {10.1145/3357713.3384334},
  year      = {2020}
}

@inproceedings{ChakrabartyL23,
  author       = {Deeparnab Chakrabarty and
                  Hang Liao},
  title        = {A Query Algorithm for Learning a Spanning Forest in Weighted Undirected
                  Graphs},
  booktitle    = {{ALT}},
  series       = {Proceedings of Machine Learning Research},
  volume       = {201},
  pages        = {259--274},
  publisher    = {{PMLR}},
  year         = {2023}
}

@inproceedings{ChakrabartyGJS23,
 author = {Chakrabarty, Deeparnab and Graur, Andrei and Jiang, Haotian and Sidford, Aaron},
 booktitle = {Advances in Neural Information Processing Systems},
 pages = {68639--68654},
 publisher = {Curran Associates, Inc.},
 title = {Parallel Submodular Function Minimization},
 url = {https://proceedings.neurips.cc/paper_files/paper/2023/file/d8a7f2f7e346410e8ac7b39d9ff28c4a-Paper-Conference.pdf},
 volume = {36},
 year = {2023}
}

@article{Jiang23,
  author       = {Haotian Jiang},
  title        = {Minimizing Convex Functions with Rational Minimizers},
  journal      = {J. {ACM}},
  volume       = {70},
  number       = {1},
  pages        = {5:1--5:27},
  year         = {2023}
}

@inproceedings{AxelrodLS20,
  author       = {Brian Axelrod and
                  Yang P. Liu and
                  Aaron Sidford},
  title        = {Near-optimal Approximate Discrete and Continuous Submodular Function
                  Minimization},
  booktitle    = {{SODA}},
  pages        = {837--853},
  publisher    = {{SIAM}},
  year         = {2020}
}

@article{Vygen03,
  author       = {Jens Vygen},
  title        = {A note on Schrijver's submodular function minimization algorithm},
  journal      = {J. Comb. Theory, Ser. {B}},
  volume       = {88},
  number       = {2},
  pages        = {399--402},
  year         = {2003}
}

@article{Orlin09,
  author       = {James B. Orlin},
  title        = {A faster strongly polynomial time algorithm for submodular function
                  minimization},
  journal      = {Math. Program.},
  volume       = {118},
  number       = {2},
  pages        = {237--251},
  year         = {2009}
}

@inproceedings{IwataO09,
  author       = {Satoru Iwata and
                  James B. Orlin},
  title        = {A simple combinatorial algorithm for submodular function minimization},
  booktitle    = {{SODA}},
  pages        = {1230--1237},
  publisher    = {{SIAM}},
  year         = {2009}
}

@article{Iwata03,
  author       = {Satoru Iwata},
  title        = {A Faster Scaling Algorithm for Minimizing Submodular Functions},
  journal      = {{SIAM} J. Comput.},
  volume       = {32},
  number       = {4},
  pages        = {833--840},
  year         = {2003}
}

@article{Schrijver00,
  author       = {Alexander Schrijver},
  title        = {A Combinatorial Algorithm Minimizing Submodular Functions in Strongly
                  Polynomial Time},
  journal      = {J. Comb. Theory, Ser. {B}},
  volume       = {80},
  number       = {2},
  pages        = {346--355},
  year         = {2000}
}

@article{IwataFF01,
  author       = {Satoru Iwata and
                  Lisa Fleischer and
                  Satoru Fujishige},
  title        = {A combinatorial strongly polynomial algorithm for minimizing submodular
                  functions},
  journal      = {J. {ACM}},
  volume       = {48},
  number       = {4},
  pages        = {761--777},
  year         = {2001}
}

@article{FleischerI03,
  author       = {Lisa Fleischer and
                  Satoru Iwata},
  title        = {A push-relabel framework for submodular function minimization and
                  applications to parametric optimization},
  journal      = {Discret. Appl. Math.},
  volume       = {131},
  number       = {2},
  pages        = {311--322},
  year         = {2003}
}

@article{Cunningham85,
  author       = {William H. Cunningham},
  title        = {On submodular function minimization},
  journal      = {Comb.},
  volume       = {5},
  number       = {3},
  pages        = {185--192},
  year         = {1985}
}

@book{GLS1988,
  author       = {Martin Gr{\"{o}}tschel and
                  L{\'{a}}szl{\'{o}} Lov{\'{a}}sz and
                  Alexander Schrijver},
  title        = {Geometric Algorithms and Combinatorial Optimization},
  series       = {Algorithms and Combinatorics},
  volume       = {2},
  publisher    = {Springer},
  year         = {1988}
}

@inproceedings{ChakrabartyLS0W19,
  author       = {Deeparnab Chakrabarty and
                  Yin Tat Lee and
                  Aaron Sidford and
                  Sahil Singla and
                  Sam Chiu{-}wai Wong},
  title        = {Faster Matroid Intersection},
  booktitle    = {{FOCS}},
  pages        = {1146--1168},
  publisher    = {{IEEE} Computer Society},
  year         = {2019}
}

@inproceedings{GraurPRW20,
  author    = {Andrei Graur and
               Tristan Pollner and
               Vidhya Ramaswamy and
               S. Matthew Weinberg},
  title     = {New Query Lower Bounds for Submodular Function Minimization},
  booktitle = {{ITCS}},
  series    = {LIPIcs},
  volume    = {151},
  pages     = {64:1--64:16},
  publisher = {Schloss Dagstuhl - Leibniz-Zentrum f{\"{u}}r Informatik},
  year      = {2020}
}

@inproceedings{Jiang21,
  author    = {Haotian Jiang},
  title     = {Minimizing Convex Functions with Integral Minimizers},
  booktitle = {{SODA}},
  pages     = {976--985},
  publisher = {{SIAM}},
  year      = {2021}
}

@inproceedings{Harvey08,
  author    = {Nicholas J. A. Harvey},
  title     = {Matroid intersection, pointer chasing, and Young's seminormal representation
               of \emph{S\({}_{\mbox{n}}\)}},
  booktitle = {{SODA}},
  pages     = {542--549},
  publisher = {{SIAM}},
  year      = {2008}
}

@article{Zelke11,
  author =        {Mariano Zelke},
  journal =       {Information Processing Letters},
  number =        {3},
  pages =         {145--150},
  title =         {Intractability of min- and max-cut in streaming
                   graphs},
  volume =        {111},
  year =          {2011},
}

@inproceedings{Rosti-etal-2007-combining,
    title = "Combining Outputs from Multiple Machine Translation Systems",
    author = "Rosti, Antti-Veikko  and
      Ayan, Necip Fazil  and
      Xiang, Bing  and
      Matsoukas, Spyros  and
      Schwartz, Richard  and
      Dorr, Bonnie",
    booktitle = "Human Language Technologies 2007: The Conference of the North {A}merican Chapter of the Association for Computational Linguistics; Proceedings of the Main Conference",
    month = apr,
    year = "2007",
    address = "Rochester, New York",
    publisher = "Association for Computational Linguistics",
    url = "https://aclanthology.org/N07-1029/",
    pages = "228--235"
}

@inproceedings{Kenneth-MordochK26SODA,
  author       = {Yotam Kenneth{-}Mordoch and
                  Robert Krauthgamer},
  editor       = {Kasper Green Larsen and
                  Barna Saha},
  title        = {All-Pairs Minimum Cut using {$\tO$}(n\({}^{\mbox{7/4}}\)) Cut Queries},
  booktitle    = {Proceedings of the 2026 Annual {ACM-SIAM} Symposium on Discrete Algorithms,
                  {SODA} 2026, Vancouver, BC, Canada, January 11-14, 2026},
  pages        = {4077--4095},
  publisher    = {{SIAM}},
  year         = {2026},
  url          = {https://doi.org/10.1137/1.9781611978971.150},
  doi          = {10.1137/1.9781611978971.150},
  bibsource    = {dblp computer science bibliography, https://dblp.org}
}

@inproceedings{Kenneth-MordochK26STOC,
  author       = {Yotam Kenneth{-}Mordoch and
                  Robert Krauthgamer},
  title        = {Faster All-Pairs Minimum Cut: Bypassing Exact Max-Flow},
  booktitle    = {Proceedings of the 58th Annual {ACM} Symposium on Theory of Computing,
                  {STOC} 2026, Salt Lake City, UT, USA, June 22-26, 2026},
  pages        = {1254--1265},
  publisher    = {{ACM}},
  year         = {2026},
  url          = {https://doi.org/10.1145/3798129.3800836},
  doi          = {10.1145/3798129.3800836}
}

@inproceedings{JiangNS26,
  author       = {Yonggang Jiang and
                  Danupon Nanongkai and
                  Pachara Sawettamalya},
  title        = {Minimum s--t Cuts with Fewer Cut Queries},
  booktitle    = {Proceedings of the 2026 Annual {ACM-SIAM} Symposium on Discrete Algorithms,
                  {SODA} 2026, Vancouver, BC, Canada, January 11-14, 2026},
  pages        = {258--296},
  publisher    = {{SIAM}},
  year         = {2026},
  url          = {https://doi.org/10.1137/1.9781611978971.12},
  doi          = {10.1137/1.9781611978971.12}
}

@inproceedings{Kenneth-MordochK25ESA,
  author       = {Yotam Kenneth{-}Mordoch and
                  Robert Krauthgamer},
  title        = {Cut-Query Algorithms with Few Rounds},
  booktitle    = {33rd Annual European Symposium on Algorithms, {ESA} 2025, Warsaw,
                  Poland, September 15-17, 2025},
  series       = {LIPIcs},
  volume       = {351},
  pages        = {100:1--100:14},
  publisher    = {Schloss Dagstuhl - Leibniz-Zentrum f{\"{u}}r Informatik},
  year         = {2025},
  url          = {https://doi.org/10.4230/LIPIcs.ESA.2025.100},
  doi          = {10.4230/LIPICS.ESA.2025.100}
}

@string(TCS = "Theor. Comput. Sci.")

@STRING{ccc	= "Annual IEEE Conference on Computational Complexity" }

@STRING{esa	= "Annual European Symposium on Algorithms" }

@STRING{focs	= "Annual IEEE Symposium on Foundations of Computer Science" }

@STRING{icalp	= "International Colloquium on Automata, Languages and
		  Programming" }

@STRING{prelim	= "Preliminary version in " }

@STRING{proc11	= "Proc. 11th " }

@STRING{proc31	= "Proc. 31st " }

@STRING{soda	= "Annual ACM-SIAM Symposium on Discrete Algorithms" }

@STRING{stoc	= "Annual ACM Symposium on the Theory of Computing" }

@STRING{tcs	= "Theor. Comput. Sci." }

@Article{	  feigenbaumkmsz05,
  author	= {Joan Feigenbaum and Sampath Kannan and Andrew McGregor and
		  Siddharth Suri and Jian Zhang},
  title		= {On graph problems in a semi-streaming model},
  journal	= tcs,
  year		= 2005,
  volume	= 348,
  number	= {2--3},
  pages		= {207--216},
  note		= prelim # "\em "  # proc31 # icalp
		# "\em\/, pages 531--543, 2004"
}

@article{AssadiMultipassSurvey,
author = {Assadi, Sepehr},
title = {Recent Advances in Multi-Pass Graph Streaming Lower Bounds},
year = {2023},
issue_date = {September 2023},
publisher = {Association for Computing Machinery},
address = {New York, NY, USA},
volume = {54},
number = {3},
issn = {0163-5700},
url = {https://doi.org/10.1145/3623800.3623808},
doi = {10.1145/3623800.3623808},
journal = {ACM SIGACT News},
month = {Sep},
pages = {48–75},
numpages = {28}
}

@InProceedings{	  yao79,
  author	= {Andrew C. Yao},
  title		= {Some complexity questions related to distributive
		  computing},
  booktitle	= proc11 # stoc,
  pages		= {209--213},
  year		= 1979
}

@article{CoppersmithFR10,
  author    = {Don Coppersmith and
               Lisa Fleischer and
               Atri Rudra},
  title     = {Ordering by weighted number of wins gives a good ranking for weighted
               tournaments},
  journal   = {{ACM} Trans. Algorithms},
  volume    = {6},
  number    = {3},
  pages     = {55:1--55:13},
  year      = {2010},
  url       = {https://doi.org/10.1145/1798596.1798608},
  doi       = {10.1145/1798596.1798608},
  bibsource = {dblp computer science bibliography, https://dblp.org}
}

@inproceedings{ChakrabartiGMV20,
  author       = {Amit Chakrabarti and
                  Prantar Ghosh and
                  Andrew McGregor and
                  Sofya Vorotnikova},
  title        = {Vertex Ordering Problems in Directed Graph Streams},
  booktitle    = {Proceedings of the 2020 {ACM-SIAM} Symposium on Discrete Algorithms,
                  {SODA} 2020, Salt Lake City, UT, USA, January 5-8, 2020},
  pages        = {1786--1802},
  publisher    = {{SIAM}},
  year         = {2020},
  doi          = {10.1137/1.9781611975994.109}
}

@inproceedings{BawejaJW22,
  author       = {Anubhav Baweja and
                  Justin Jia and
                  David P. Woodruff},
  title        = {An Efficient Semi-Streaming {PTAS} for Tournament Feedback Arc Set
                  with Few Passes},
  booktitle    = {13th Innovations in Theoretical Computer Science Conference, {ITCS}
                  2022, January 31 - February 3, 2022, Berkeley, CA, {USA}},
  series       = {LIPIcs},
  volume       = {215},
  pages        = {16:1--16:23},
  publisher    = {Schloss Dagstuhl - Leibniz-Zentrum f{\"{u}}r Informatik},
  year         = {2022},
  doi          = {10.4230/LIPIcs.ITCS.2022.16}
}

@inproceedings{AssadiD21,
  author       = {Sepehr Assadi and
                  Aditi Dudeja},
  title        = {A Simple Semi-Streaming Algorithm for Global Minimum Cuts},
  booktitle    = {4th Symposium on Simplicity in Algorithms, {SOSA} 2021, Virtual Conference,
                  January 11-12, 2021},
  pages        = {172--180},
  publisher    = {{SIAM}},
  year         = {2021},
  url          = {https://doi.org/10.1137/1.9781611976496.19},
  doi          = {10.1137/1.9781611976496.19},
  bibsource    = {dblp computer science bibliography, https://dblp.org}
}

@inproceedings{GuruswamiO13,
  author       = {Venkatesan Guruswami and
                  Krzysztof Onak},
  title        = {Superlinear Lower Bounds for Multipass Graph Processing},
  booktitle    = {Proceedings of the 28th Conference on Computational Complexity, {CCC}
                  2013, K.lo Alto, California, USA, 5-7 June, 2013},
  pages        = {287--298},
  publisher    = {{IEEE} Computer Society},
  year         = {2013},
  url          = {https://doi.org/10.1109/CCC.2013.37},
  doi          = {10.1109/CCC.2013.37}
}

@inproceedings{AssadiR20,
  author       = {Sepehr Assadi and
                  Ran Raz},
  title        = {Near-Quadratic Lower Bounds for Two-Pass Graph Streaming Algorithms},
  booktitle    = {61st {IEEE} Annual Symposium on Foundations of Computer Science, {FOCS}
                  2020, Durham, NC, USA, November 16-19, 2020},
  pages        = {342--353},
  publisher    = {{IEEE}},
  year         = {2020},
  url          = {https://doi.org/10.1109/FOCS46700.2020.00040},
  doi          = {10.1109/FOCS46700.2020.00040}
}

@inproceedings{ChenKPSSY21-lognpass,
  author       = {Lijie Chen and
                  Gillat Kol and
                  Dmitry Paramonov and
                  Raghuvansh R. Saxena and
                  Zhao Song and
                  Huacheng Yu},
  title        = {Almost optimal super-constant-pass streaming lower bounds for reachability},
  booktitle    = {{STOC} '21: 53rd Annual {ACM} {SIGACT} Symposium on Theory of Computing,
                  Virtual Event, Italy, June 21-25, 2021},
  pages        = {570--583},
  publisher    = {{ACM}},
  year         = {2021},
  url          = {https://doi.org/10.1145/3406325.3451038},
  doi          = {10.1145/3406325.3451038}
}

@InProceedings{LauraS11,
author="Laura, Luigi
and Santaroni, Federico",
title="Computing Strongly Connected Components in the Streaming Model",
booktitle="Theory and Practice of Algorithms in (Computer) Systems",
year="2011",
publisher="Springer Berlin Heidelberg",
pages="193--205"
}

@book{Moon1968TopicsOT,
author = {Moon, John W.},
title = {Topics on tournaments},
publisher = {Holt, Rinehart and Winston},
series = {Athena series; selected topics in mathematics},
year = {1968},
}

@inproceedings{GengLQALS08,
author = {Geng, Xiubo and Liu, Tie-Yan and Qin, Tao and Arnold, Andrew and Li, Hang and Shum, Heung-Yeung},
title = {Query dependent ranking using K-nearest neighbor},
year = {2008},
publisher = {Association for Computing Machinery},
doi = {10.1145/1390334.1390356},
booktitle = {Proceedings of the 31st Annual International ACM SIGIR Conference on Research and Development in Information Retrieval},
pages = {115–122},
numpages = {8},
series = {SIGIR '08}
}

@inproceedings{MandePSS24,
  author       = {Nikhil S. Mande and
                  Manaswi Paraashar and
                  Swagato Sanyal and
                  Nitin Saurabh},
  title        = {On the Communication Complexity of Finding a King in a Tournament},
  booktitle    = {Approximation, Randomization, and Combinatorial Optimization. Algorithms
                  and Techniques, {APPROX/RANDOM} 2024, August 28-30, 2024, London School
                  of Economics, London, {UK}},
  series       = {LIPIcs},
  volume       = {317},
  pages        = {64:1--64:23},
  publisher    = {Schloss Dagstuhl - Leibniz-Zentrum f{\"{u}}r Informatik},
  year         = {2024},
  url          = {https://doi.org/10.4230/LIPIcs.APPROX/RANDOM.2024.64},
  doi          = {10.4230/LIPICS.APPROX/RANDOM.2024.64}
}

@Article{GilbertI10,
  author    = {Anna C. Gilbert and Piotr Indyk},
  title     = {Sparse Recovery Using Sparse Matrices},
  journal   = {Proceedings of the {IEEE}},
  volume    = 98,
  number    = 6,
  pages     = {937--947},
  year      = 2010
}

@INPROCEEDINGS{DasV13,
  author={Das, Abhik Kumar and Vishwanath, Sriram},
  booktitle={2013 IEEE International Conference on Acoustics, Speech and Signal Processing}, 
  title={On finite alphabet compressive sensing}, 
  year={2013},
  volume={},
  number={},
  pages={5890-5894},
  doi={10.1109/ICASSP.2013.6638794}}

@article{GassTournaments,
author = {S I Gass},
title = {Tournaments, transitivity and pairwise comparison matrices},
journal = {Journal of the Operational Research Society},
volume = {49},
number = {6},
pages = {616--624},
year = {1998},
publisher = {Taylor \& Francis},
doi = {10.1057/palgrave.jors.2600572}}

@inproceedings{GhoshK24,
  author       = {Prantar Ghosh and
                  Sahil Kuchlous},
  title        = {New Algorithms and Lower Bounds for Streaming Tournaments},
  booktitle    = {32nd Annual European Symposium on Algorithms, {ESA} 2024, September
                  2-4, 2024, Royal Holloway, London, United Kingdom},
  series       = {LIPIcs},
  volume       = {308},
  pages        = {60:1--60:19},
  publisher    = {Schloss Dagstuhl - Leibniz-Zentrum f{\"{u}}r Informatik},
  year         = {2024},
  url          = {https://doi.org/10.4230/LIPIcs.ESA.2024.60},
  doi          = {10.4230/LIPICS.ESA.2024.60}
}

@inproceedings{BarberoPP17,
  author       = {Florian Barbero and
                  Christophe Paul and
                  Michal Pilipczuk},
  title        = {Exploring the Complexity of Layout Parameters in Tournaments and Semi-Complete
                  Digraphs},
  booktitle    = {44th International Colloquium on Automata, Languages, and Programming,
                  {ICALP} 2017, July 10-14, 2017, Warsaw, Poland},
  series       = {LIPIcs},
  volume       = {80},
  pages        = {70:1--70:13},
  publisher    = {Schloss Dagstuhl - Leibniz-Zentrum f{\"{u}}r Informatik},
  year         = {2017},
  url          = {https://doi.org/10.4230/LIPIcs.ICALP.2017.70},
  doi          = {10.4230/LIPICS.ICALP.2017.70}
}

@phdthesis{Fradkin-thesis11,
author = {Fradkin, Alexandra Ovetsky},
advisor = {Chudnovsky, Maria and Seymour, Paul},
title = {Forbidden structures and algorithms in graphs and digraphs},
year = {2011},
isbn = {9781124735511},
school = {Princeton University}
}

@book{fleischner1990eulerian,
  title={Eulerian Graphs and Related Topics},
  author={Fleischner, H.},
  number={pt. 1, v. 2},
  lccn={90007300},
  series={Annals of discrete mathematics},
  url={https://books.google.com/books?id=6Df1xQEACAAJ},
  year={1990},
  publisher={North-Holland}
}

@InProceedings{fp13,
author="Fomin, Fedor V.
and Pilipczuk, Micha{\l}",
editor="Bodlaender, Hans L.
and Italiano, Giuseppe F.",
title="Subexponential Parameterized Algorithm for Computing the Cutwidth of a Semi-complete Digraph",
booktitle="Algorithms -- ESA 2013",
year="2013",
publisher="Springer Berlin Heidelberg",
address="Berlin, Heidelberg",
pages="505--516"
}

@InProceedings{bfis25,
  author =	{Bentert, Matthias and Fomin, Fedor V. and Inamdar, Tanmay and Saurabh, Saket},
  title =	{{Exponential-Time Approximation (Schemes) for Vertex-Ordering Problems}},
  booktitle =	{16th Innovations in Theoretical Computer Science Conference (ITCS 2025)},
  pages =	{15:1--15:18},
  series =	{Leibniz International Proceedings in Informatics (LIPIcs)},
  ISBN =	{978-3-95977-361-4},
  ISSN =	{1868-8969},
  year =	{2025},
  volume =	{325},
  editor =	{Meka, Raghu},
  publisher =	{Schloss Dagstuhl -- Leibniz-Zentrum f{\"u}r Informatik},
  address =	{Dagstuhl, Germany},
  URL =		{https://drops.dagstuhl.de/entities/document/10.4230/LIPIcs.ITCS.2025.15},
  URN =		{urn:nbn:de:0030-drops-226431},
  doi =		{10.4230/LIPIcs.ITCS.2025.15}
}

@misc{devos2025maximumlineararrangementproblem,
      title={A Maximum Linear Arrangement Problem on Directed Graphs}, 
      author={Matt DeVos and Kathryn Nurse},
      year={2025},
      eprint={1810.12277},
      archivePrefix={arXiv},
      primaryClass={cs.DS},
      url={https://arxiv.org/abs/1810.12277}, 
}

@inproceedings{ClaessonDF23,
author = {Claesson, Anders and Dukes, Mark and Franklín, Atli and Stefánsson, Sigurður},
year = {2023},
month = {01},
pages = {290-297},
title = {Counting tournament score sequences},
booktitle = {Proceedings of the American Mathematical Society},
doi = {10.5817/CZ.MUNI.EUROCOMB23-040}
}

@article{Narayana_Bent_1964, 
title={Computation of the Number of Score Sequences in Round-Robin Tournaments}, 
volume={7}, 
DOI={10.4153/CMB-1964-015-1},
number={1}, 
journal={Canadian Mathematical Bulletin}, 
author={Narayana, T.V. and Bent, D.H.}, 
year={1964}, 
pages={133–136}
}

@article{Moon_1962, 
title={On the Score Sequence of an N-Partite Tournament}, 
volume={5}, 
DOI={10.4153/CMB-1962-008-9},
number={1}, 
journal={Canadian Mathematical Bulletin}, 
author={Moon, J.W.}, 
year={1962}, 
pages={51–58}
}

@article{RIORDAN1968,
title = {The number of score sequences in tournaments},
journal = {Journal of Combinatorial Theory},
volume = {5},
number = {1},
pages = {87-89},
year = {1968},
issn = {0021-9800},
doi = {https://doi.org/10.1016/S0021-9800(68)80032-8},
url = {https://www.sciencedirect.com/science/article/pii/S0021980068800328},
author = {John Riordan}
}

@article{Kleitman1981ForestsAS,
  title={Forests and score vectors},
  author={Daniel J. Kleitman and Kenneth J. Winston},
  journal={Combinatorica},
  year={1981},
  volume={1},
  pages={49-54},
  url={https://api.semanticscholar.org/CorpusID:7639560}
}

@misc{stockmeyer2023counting,
      title={Counting Various Classes of Tournament Score Sequences}, 
      author={Paul K. Stockmeyer},
      year={2023},
      eprint={2202.05238},
      archivePrefix={arXiv},
      primaryClass={math.CO},
      url={https://arxiv.org/abs/2202.05238}, 
}

@inproceedings{Brand0PKLGSS23,
  author       = {Jan van den Brand and
                  Li Chen and
                  Richard Peng and
                  Rasmus Kyng and
                  Yang P. Liu and
                  Maximilian Probst Gutenberg and
                  Sushant Sachdeva and
                  Aaron Sidford},
  title        = {A Deterministic Almost-Linear Time Algorithm for Minimum-Cost Flow},
  booktitle    = {64th {IEEE} Annual Symposium on Foundations of Computer Science, {FOCS}
                  2023, Santa Cruz, CA, USA, November 6-9, 2023},
  pages        = {503--514},
  publisher    = {{IEEE}},
  year         = {2023},
  url          = {https://doi.org/10.1109/FOCS57990.2023.00037},
  doi          = {10.1109/FOCS57990.2023.00037}
}

@article{Cun85,
author = {Cunningham, William H.},
title = {Minimum cuts, modular functions, and matroid polyhedra},
journal = {Networks},
volume = {15},
number = {2},
pages = {205-215},
doi = {https://doi.org/10.1002/net.3230150206},
url = {https://onlinelibrary.wiley.com/doi/abs/10.1002/net.3230150206},
eprint = {https://onlinelibrary.wiley.com/doi/pdf/10.1002/net.3230150206},
year = {1985}
}

@article{FrankGyarfas1978,
  title = {How to orient the edges of a graph?},
  author = {Frank, Andr\'{a}s and Gy\'{a}rf\'{a}s, Andr\'{a}s},
  journal = {Combinatorics},
  series = {Colloquia Mathematica Societatis János Bolyai},
  volume = {18},
  pages = {353--364},
  year = {1978},
  publisher = {North-Holland},
  address = {Amsterdam}
}

@article{ChenCLT26Independence,
  author       = {Ho{-}Lin Chen and
                  Tsun Ming Cheung and
                  Peng{-}Ting Lin and
                  Meng{-}Tsung Tsai},
  title        = {Independence-Number Parameterized Space Complexity for Directed Connectivity
                  Certificate},
  journal      = {CoRR},
  volume       = {abs/2602.12668},
  year         = {2026},
  doi          = {10.48550/ARXIV.2602.12668}
}

@article{ChakrabartyCK22,
  author={Deeparnab Chakrabarty and Yu Chen and Sanjeev Khanna}, 
  title={Notes on Directed Cut Oracle}, 
  journal={Unpublished},
  year={2022}
  }
\bibliographystyle{alpha}

    \end{document}